\documentclass[10pt, conference, letterpaper]{IEEEtran}
\IEEEoverridecommandlockouts

\usepackage[margin=1in]{geometry}
\usepackage[bookmarks, colorlinks=true, plainpages = false, citecolor = blue,linkcolor=red,urlcolor = blue, filecolor = blue]{hyperref}
\usepackage{url}

\usepackage{amsmath,amsfonts,amssymb,bm, verbatim,mathtools,amsthm}
\usepackage[linesnumbered,ruled]{algorithm2e}

\usepackage[usenames, dvipsnames]{color}

\usepackage{etoolbox}
\usepackage{tikz}
\usepackage{xr,xspace}
\usepackage{todonotes}
\usepackage{paralist}
\usepackage{caption,bbm}

\newtheorem{definition}{Definition}
\newtheorem{lemma}{Lemma}

\newtheorem{theorem}{Theorem}

\newtheorem{assumption}{Assumption}

\newcommand{\mycomment}[1]{}

\usepackage{xspace,prettyref}

\newcommand{\toas}{\xrightarrow{{\rm a.s.}}}

\newrefformat{eq}{(\ref{#1})}
\newrefformat{chap}{Chapter~\ref{#1}}
\newrefformat{sec}{Section~\ref{#1}}
\newrefformat{algo}{Algorithm~\ref{#1}}
\newrefformat{fig}{Fig.~\ref{#1}}
\newrefformat{tab}{Table~\ref{#1}}
\newrefformat{rmk}{Remark~\ref{#1}}
\newrefformat{clm}{Claim~\ref{#1}}
\newrefformat{def}{Definition~\ref{#1}}
\newrefformat{cor}{Corollary~\ref{#1}}
\newrefformat{lmm}{Lemma~\ref{#1}}
\newrefformat{prop}{Proposition~\ref{#1}}
\newrefformat{app}{Appendix~\ref{#1}}
\newrefformat{hyp}{Hypothesis~\ref{#1}}
\newrefformat{problem}{Problem~\ref{#1}}
\newrefformat{thm}{Theorem~\ref{#1}}

\newcommand{\pth}[1]{\left( #1 \right)}

\newcommand{\calD}{{\mathcal{D}}}
\newcommand{\calE}{{\mathcal{E}}}
\newcommand{\calF}{{\mathcal{F}}}

\newcommand{\calH}{{\mathcal{H}}}

\newcommand{\calL}{{\mathcal{L}}}

\newcommand{\calN}{{\mathcal{N}}}

\newcommand{\calR}{{\mathcal{R}}}
\newcommand{\calS}{{\mathcal{S}}}

\newcommand{\calV}{{\mathcal{V}}}

\begin{document}

\title{Distributed Learning with Adversarial Agents Under Relaxed Network Condition\\
	 \thanks{Research reported in this paper was sponsored in part by the Army Research Laboratory under Cooperative Agreement W911NF-17-2-0196, and by National Science Foundation awards 1421918 and 1610543. The views and conclusions contained in this document are those of the authors and should not be interpreted as representing the official policies, either expressed or implied, of the the Army Research Laboratory, National Science Foundation or the U.S. Government. The U.S. Government is authorized to reproduce and distribute reprints for Government purposes notwithstanding any copyright notation here on.}}

\author{\IEEEauthorblockN{Pooja Vyavahare\IEEEauthorrefmark{1}, Lili Su\IEEEauthorrefmark{2}, Nitin H. Vaidya\IEEEauthorrefmark{3}}
\IEEEauthorblockA{\IEEEauthorrefmark{1}Coordinated Science Laboratory, University of Illinois at Urbana-Champaign\\
Email: poojav@illinois.edu}
\IEEEauthorblockA{\IEEEauthorrefmark{2}Computer Science and Artificial Intelligence Laboratory,	Massachusetts Institute of Technology\\
	Email: lilisu@mit.edu}
\IEEEauthorblockA{\IEEEauthorrefmark{3}Department of Computer Science, Georgetown University\\
	Email: nitin.vaidya@georgetown.edu}
}	

\maketitle

\begin{abstract}
	This work studies the problem of non-Bayesian learning over multi-agent network when there are some adversarial (faulty) agents in the network. At each time step, each non-faulty agent collects partial information about an unknown state of the world and tries to estimate true state of the world by iteratively sharing information with its neighbors. Existing algorithms in this setting require that all non-faulty agents in the network should be able to achieve consensus via local information exchange.

In this work, we present an analysis of a distributed algorithm which does not require the network to achieve consensus. We show that if every non-faulty agent can receive enough information (via iteratively communicating with neighbors) to differentiate the true state of the world from other possible states then it can indeed learn the true state. 

\end{abstract}

\section{Introduction}
\label{sec:intro}

Distributed algorithms in multi-agent networks for various network settings have been studied since long time \cite{Tsitsiklis84,Gallager88}.
In this work, we consider a set of agents which are connected by directed links, thus forming a directed network. Each agent is attached to a sensor which senses some partial information about the state of the world (environment) in which the network is present. There is only one true state of the world and the aim for each agent is to estimate the true state by iteratively sharing information with its neighbors. \emph{Distributed learning} has been studied in different settings like in the presence of a fusion center \cite{Varshney12,Tsitsiklis93} and when there is no fusion center \cite{Gale03,Cattivelli11,Jakovetic12}. 

Non-Bayesian learning with the use of iterative distributed consensus algorithm was first proposed by Jadbabaie et. al. \cite{Jadbabaie12}. The approach proposed in \cite{Jadbabaie12} requires the network formed by the agents to achieve consensus in order to learn the true state. Since then the non-Bayesian learning has been applied in various network settings; see \cite{Nedic16} for a survey of results in this area.

%Our aim is to study a network of agents in which an unknown set of agents is adversarial. We assume that an adversarial agent suffers Byzantine faults, i.e., it may not follow the specified algorithm. An agent with Byzantine faults may send arbitrary information to its neighbors. We also assume that a faulty agent has full knowledge of the network and the algorithm specifications of all non-faulty agents. Learning the true state of the environment by non-faulty agents in a network with adversarial agents was first studied by \cite{Su16a}. Authors in \cite{Su16a} propose an algorithm in which every agent updates its belief by geometric average of its local Bayesian update and its neighbors' beliefs which has been studied in many works; e.g., \cite{Nedic15,Shahrampour13}.
%A Bayesian update rule is applied to learning the true state in \cite{Gale03} which is impractical in many network applications due to its high requirement of memory. 

Our aim is to study a network of agents in which an unknown set of agents is adversarial. We assume that an adversarial agent suffers Byzantine faults, i.e., may send arbitrary information to its neighbors and may not follow the specified algorithm. Learning the true state of the world in a network with adversarial agents was first studied by \cite{Su16a,Su16c}. The algorithm in \cite{Su16a} uses a geometric averaging update similar to that used in other works \cite{Nedic15,Shahrampour13}.

The algorithm analysis in \cite{Su16a,Su16c} requires that the network topology be such that non-faulty agents can achieve consensus by iteratively sharing their information with their neighbors. In this work we circumvent this limitation. We analyze the algorithm proposed in \cite{Su16a,Su16c} and show that in order to estimate the true state of the world by non-faulty agents, achieving distributed consensus is not required. Intuitively, we show that if the set of agents that can reach an agent can collectively estimate the true state then the agent can also estimate the true state almost surely. 

%Existing work in this area use a non-Bayesian learning approach which combines local Bayesian information with distributed consensus.
% In this setting, each agent updates its \emph{belief} of the true state in each iteration.
%More specifically if the set of agents which can communicate with an agent can collectively estimate the true state of the world then the agent can also estimate the true state. We analyze a distributed algorithm in which agents iteratively collect partial information about the state of the world and collaboratively try to estimate the true state. Our analysis shows that for all fault-free agents to estimate the true state, consensus among agents is not required.
% We also assume that a faulty agent has full knowledge of the network and the algorithm specifications of all non-faulty agents. 
%In this work we consider a synchronous setting, i.e., all agents send their information to their neighbors at the same time in a round.  

%distributed hypothesis learning 
%synchronous setting
%static network not time-varying networks
%non-Bayesian learning
%adversarial agents

\subsection{Preview}
We introduce the system model in Section~\ref{sec:formulation} and present the algorithm to estimate true state in presence of adversarial agents (which is first introduced in \cite{Su16a}) in Section~\ref{sec:faulty}. In Section~\ref{sec:f_matrix_property} we state our assumption on network along with our main contribution (Lemma~\ref{lm:f_matrix}). We use this lemma to analyze Algorithm~\ref{alg:fault} in Section~\ref{sec:f_alg_analysis}. We conclude the work in Section~\ref{sec:conclusion}.

%The outline of the paper and contributions of the work are as outlined below. We introduce the system model and basic notations in Section~\ref{sec:formulation}. In Section~\ref{sec:faulty} we present the algorithm to estimate true state in presence of adversarial agents and notations required for its analysis in Section~\ref{sec:f_matrix_property}. We present the analysis of the algorithm and main contribution of this work in Section~\ref{sec:f_alg_analysis}. Note that the algorithm presented in this work was first proposed in \cite{Su16a} but we present a novel analysis of the algorithm which gets around the need of achieving consensus in order to estimate the true state of the environment. We extend the analysis to a network with no adversaries in Section~\ref{sec:faulty} and conclude the work with possible future works in Section~\ref{sec:conclusion}.
\section{Problem Formulation}
\label{sec:formulation}

 We consider a system model similar to that in \cite{Vaidya14,Vaidya12}. We consider a set of agents which are connected via directed links thus forming a directed network $G=(\calV,\calE)$ where $|\calV| =n.$ We consider synchronous system setting. Maximum $f$ agents can suffer with Byzantine faults at each execution of the algorithm. Any agent with Byzantine fault may send arbitrary different information to different neighbors. Adversarial agents can collaborate with each other and have full knowledge of the system. Let $\calF$ be the set of faulty agents and $\calN$ be set of non-faulty (good) agents in an execution. Each good agent at every execution knows the upper bound on number of faulty agents, i.e., $f,$ but does not know the set $\calF.$ Let $|\calF|=\phi.$
 
 Every agent collects some partial information about the world. The aim of every fault-free agent is to estimate the true state of the world by iteratively sharing information with neighbors. For this we use the same model as presented in \cite{Jadbabaie12,Su16c}. Let there be $m$ possible states of the world and we represent the set of states by: $\Theta = \{\theta_1,\theta_2, \ldots,\theta_m\}.$ Out of $m$ possible states, there is one true state $\theta^* \in \Theta.$ Initially, at $t=0,$ the true state is unknown to every agent in the network. At every time iteration $t,$ every agent independently observes some information (signal) about the state $\theta^*.$ Observed signal space for agent $i$ is represented by $\calS_i$ and we assume that $|\calS_i| < \infty.$  Let $\ell_i(.|\theta)$ be the marginal distribution of the signal observed by agent $i$ when the true state is $\theta.$ Each agent $i$ knows the structure of its observed signal which is represented by a set of marginal distributions $\calD_i = \{\ell_i(\omega_i|\theta)| \theta \in \Theta, \omega_i \in \calS_i \}.$ We also assume that $\forall \omega_i \in \calS_i$ and $\forall \theta \in \Theta,$ $\ell_i(\omega_i|\theta) >0.$ In other words, the support of the distribution $\ell_i(.|\theta)$ is the whole signal space. Let $s_{1,t}^i$ be the signal history observed by agent $i$ up to time $t.$ Throughout this work, log of any vector ${\bf x}$ is defined as a vector ${\bf y}$ with ${\bf y}[i] := \log ({\bf x}[i]),$ i.e., the log operation on a vector is element-wise.

%Right now we consider that there are no faulty agents. For each $i \in \calV,$ let $\calI_i$ be the set of incoming neighbors of agent $i.$
%In each execution of the algorithm there can be up to $f$ adversarial agents . 
%
\section{Non-Bayesian learning with faulty agents}
\label{sec:faulty}

In this section we present the algorithm for non-Bayesian learning when some agents in the network are faulty. Note that Algorithm~\ref{alg:fault} was first presented in \cite{Su16a} and in this work we present an improved analysis which circumvent the need to achieve consensus in order to learn the true state by non faulty agents. Algorithm~\ref{alg:fault} and some related concepts are presented here for the sake of completeness of this manuscript. For more details refer to \cite{Su16a,Su18a,Vaidya14}.

 For convenience of presentation, we assume that the non-faulty agents are numbered $1,2,\ldots,n-\phi$ (where $\phi=|\calF|$ is the number of faulty agents at every time iteration). At each time iteration $t,$ every non-faulty agent $i$ maintains a vector $\mu_t^i \in \mathbb{R}^m$ of the possible states of the world. $\mu_t^i$ is a stochastic vector over all states $\theta \in \Theta$ with $ 0 \leq\mu_t^i(\theta)\leq 1$ and $\sum_{\theta} \mu_t^i(\theta) =1 \forall i.$ We assume that initially at $t=0,$ $\mu_0^i(\theta) = 1/m ~\forall i, ~\forall \theta \in \Theta.$

\begin{algorithm}

	\caption{\cite{Su16a} Non-Bayesian learning with faulty agents: for agent $i$}
	\label{alg:fault}
	{\normalsize
		\vskip 0.2\baselineskip
		$Z^i\gets \emptyset$\;		
		${\bf x}^i\gets \log \mu_{t-1}^i;$ \textcolor{blue}{\textit{$({\bf x}^i$ is a vector over all states with ${\bf x}^i(\theta) := \log \mu_{t-1}^i(\theta) ~\forall \theta \in \Theta)$}}\
		
		Transmit ${\bf x}^i$ on all outgoing links.\;
		\vskip 0.2\baselineskip
		
		Receive messages on all incoming links. Let these multiset of messages be $R^i.$
		\vskip 0.2\baselineskip
		
		\For {every $C\subseteq R^i\cup \{{\bf x}^i\}$ such that $|C|=(m+1)f+1$}
		{add to $Z^i$ a {\em Tverberg point} of multiset $C$}
		\vskip 0.2\baselineskip
		
		$\eta_t^i\gets \frac{1}{1+|Z^i|} \pth{{\bf x}^i+\sum_{{\bf z}\in Z^i}{\bf z}}$\;
		\vskip 0.2\baselineskip
		
		Observe $s_t^i$\;
		
		\For{$\theta\in\Theta$}
		%$\forall \theta\in\Theta$,
		{$\ell_i(s_{1,t}^i|\theta)\gets \ell_i(s^i_t|\theta)\, \ell_i(s_{1,t-1}^i|\theta)$\;
			$\mu_{t}^i(\theta)\gets \frac{\ell_i(s_{1, t}^i|\theta)\exp \pth{\eta_{t}^i(\theta)}}{\sum_{p=1}^m \ell_i(s_{1, t}^i|\theta_p)\exp \pth{\eta_{t}^i(\theta_p)}}$\;}
	
	}
	
\end{algorithm}
The \emph{Tverberg point} is guaranteed to be in the convex hull of values received from non-faulty agents. See \cite{Vaidya14} for definition of \emph{Tverberg point}. As shown in \cite{Su16a}, the dynamics of $\eta_t^i$ for fault free agent $i$ ($1\leq i\leq n-\phi$) of Algorithm~\ref{alg:fault} can be written as:
\begin{equation}
\eta_t^i(\theta) =\log \prod_{j=1}^{n-\phi}\mu_{t-1}^j(\theta)^{{\bf A}_{ij}[t]}, ~~~\forall \theta\in \Theta,
\label{eq:eta}
\end{equation}
where ${\bf A}[t]$ is a $(n-\phi) \times (n-\phi)$ row stochastic matrix corresponding to the execution of Algorithm~\ref{alg:fault} at time $t.$ As shown in \cite{Vaidya14}, ${\bf A}[t]$ is affected by the behavior of faulty agents. For any $\theta_1,\theta_2 \in \Theta,$ and for any agent $i \in \calV,$ let $\psi_t^i(\theta_1,\theta_2)$ and $\calL_t(\theta_1,\theta_2)$ be as follows:
\begin{equation}
\bm{\psi}_{t}^i(\theta_1, \theta_2)\triangleq \log \frac{\mu_t^i(\theta_1)}{\mu_t^i(\theta_2)}, \quad \calL^i_{t}(\theta_1, \theta_2)~\triangleq~\log \frac{\ell_i(s_{t}^i|\theta_1)}{\ell_i(s_{t}^i|\theta_2)}.
\label{eq:f_psi_L}
\end{equation}
Following the analysis in \cite{Su16a} the evolution of $\bm{\psi}_{t}^{i}(\theta, \theta^*)$ can be written as:
\begin{equation}
\psi_t^{i}(\theta,\theta^*) = \sum_{r=1}^{t} \sum_{j=1}^{n-\phi} {\bf \Phi}_{ij}(t,r+1) \sum_{k=1}^{r} \calL_k^j(\theta,\theta^*). 
\label{eq:f_psi_i}
\end{equation}
where ${\bf \Phi}_{ij}(t,r)$ is $(i,j)$-th element of ${\bf \Phi}(t,r) = {\bf A}[t] \ldots {\bf A}[r]$  for $1 \leq r \leq t+1.$ By convention, ${\bf \Phi}(t,t)={\bf A}[t]$ and ${\bf \Phi}(t,t+1) = {\bf I}.$

\subsection{Properties of ${\bf \Phi}(t,r)$}
	\label{sec:f_matrix_property}

Many concepts of this section were presented in \cite{Vaidya12,Su16a} and we present them here for the sake of completeness of this manuscript. Recall that ${\bf A}[t]$ is a row stochastic matrix which defines the run of Algorithm~\ref{alg:fault} at time $t.$ Note that Algorithm~\ref{alg:fault} uses Tverberg points to generate $\eta_t^i$ which is obtained by rejecting extreme values received from neighbors. It is shown in \cite{Vaidya14} that this can be seen as removing some incoming links at each round of the algorithm and the effective network can be characterized by \emph{reduced graph} of $G(\calV,\calE).$
%This effectively removes the messages from faulty agents.

\begin{definition} \cite{Vaidya14}
	A reduced graph $\calH(\calN,\calE_\calF)$ of network $G(\calV,\calE)$ is obtained by:
	\begin{enumerate}
		\item removing all faulty agents $\calF$ and all the links incident on the these agents
		\item for all non-faulty agents, removing up to $mf$ additional incoming links.
	\end{enumerate}
\label{def:reduced_graph}
\end{definition} 

Let the set of all such reduced graphs be $\calR_\calF.$ By the definition of reduced graph and finiteness of $G$ note that the number of possible reduced graphs of $G$ is finite, i.e., $|\calR_\calF|=r_f < \infty.$ A \emph{source component} in a reduced graph is a strongly connected set of agents which does not have any incoming links from outside that set. We make the following assumption in our analysis.

\begin{assumption}
	Every reduced graph contains one or more source components and each agent in the reduced graph is either a part of a source component or has a directed path from one or more source components.
	\label{as:f_multiple_source}
\end{assumption}

\textbf{Remark:} Note that analysis in \cite{Su16a} assumes that every reduced graph contains only one source component. This assumption is shown \cite{Su16a,Vaidya14} to be sufficient to achieve approximate Byzantine vector consensus. We do not assume that there is a unique source component in each reduced graph. Thus, under our assumption, consensus on arbitrary inputs is not necessarily guaranteed. However, under Assumption~\ref{as:fault_source} stated below regarding the sensor observations, the learning problem is solvable.

Assumption~\ref{as:f_multiple_source} is different than the one made in \cite{Su16a} and thus is not sufficient to achieve consensus among fault free agents. The key contribution of this work is to show correctness of Algorithm~\ref{alg:fault} under Assumption~\ref{as:f_multiple_source}.

It was shown in \cite{Vaidya12} that for any ${\bf A}[t]$ there exists a reduced graph of $G,$ say $\calH[t]$ whose transition matrix is ${\bf H}[t],$ such that ${\bf A}[t] \geq \beta {\bf H}[t]$ where $0<\beta<1$ is a constant.  
For more details on this relationship and definition of $\beta$ refer to \cite{Vaidya12}. Now we present a new result which will be used for the analysis.

\begin{lemma}
	For ${\bf \Phi}(t,r+1),$ with $t-r\geq\nu:=r_f(n-\phi),$ there exists a reduced graph $\calH_r$ such that the following holds for each $i,$ $1 \leq i \leq n-\phi:$ there exists a source component $P_r^i \in \calH_r$ such that 
	${\bf \Phi}_{ij}(t,r+1) \geq \beta^\nu/n$ for each agent $j$ in that source component of $\calH_r.$
	
	\label{lm:f_matrix}
\end{lemma}
\begin{proof}
	We will prove the result for two cases. First for $t-r = \nu,$ recall the product matrix ${\bf \Phi}(t,r+1) = {\bf A}[t] \ldots {\bf A}[r+1]$ and for any ${\bf A}[x] \geq \beta {\bf H}[x]$ where ${\bf H}[x]$ is the adjacency matrix of the reduced graph corresponding to $x$-th round of Algorithm~\ref{alg:fault}. Thus,
	\begin{equation*}
	   {\bf \Phi}(t,r+1) \geq \beta^{\nu} \prod_{x=r+1}^{t} {\bf H}[x].
	\end{equation*}
	The product ${\bf \Phi}(t,r+1)$ contains $\nu=r_f(n-\phi)$ reduced graphs of $G.$ As there are $r_f$ distinct reduced graphs, there is one reduced graph $\calH_r$ which will occur at least $(n-\phi)$ times in ${\bf \Phi}(t,r+1).$ By Assumption~\ref{as:f_multiple_source}, every agent has a directed path from at least one source component in $\calH_r$ and let $P_r^i$ be any one source component which has a directed path to $i$ in $\calH_r.$ As the maximum length of any path in $\calH_r$ is $(n-\phi-1),$ for each agent $i,$ $(\prod_{x=r+1}^{t} {\bf H}[x])_{ij} \geq 1$ for all $j \in P_r^i.$ Thus for each agent $i,$ and $j \in P_r^i,$ ${\bf \Phi}_{ij}(t,r+1) \geq \beta^{\nu}>\beta^{\nu}/n.$ Hence the result is proved when $t-r=\nu.$
	
	Now, for any value of $t,r$ such that $t-r =\nu+k$ where $k \geq 1$ is an integer, we get
	\begin{align*}
	{\bf \Phi}(t,r+1) &= {\bf A}[t]\ldots{\bf A}[t-k+1] {\bf A}[t-k] \ldots {\bf A}[r+1] \\
	&= {\bf \Phi}(t,t-k+2) {\bf \Phi}(t-k+1,r+1).
	\end{align*}
	Let the $i$-th row of ${\bf \Phi}(t-k+1,r+1)$ be $K_i$ and that of ${\bf \Phi}(t,r+1)$ be $L_i.$ Then $L_i$ can be written in terms of $K_i$ as:
	\begin{equation*}
	L_i = \sum_{j=1}^{n-\phi} {\bf \Phi}_{ij}(t,t-k+2) K_j.
	\end{equation*}
	Recall that ${\bf \Phi}(t,t-k+2)$ is a $(n-\phi) \times (n-\phi)$ row stochastic matrix thus for every $i,$ there exists some $j$ such that ${\bf \Phi}_{ij}(t,t-k+2) \geq 1/(n-\phi) \geq 1/n.$
	By first part of the proof, there exists a reduced graph $\calH_r$ such that for each row $j$ of ${\bf \Phi}(t-k+1,r+1)$ there exists a source component of $\calH_r$ such that ${\bf \Phi}_{jp}(t-k+1,r+1)\geq \beta^\nu$ where $p$ belongs to that source component. Thus, for each row $L_i$ of ${\bf \Phi}(t,r+1)$ there exists a reduced graph $\calH_r$ such that there exists a source component $P_r^i$ of $\calH_r$ such that ${\bf \Phi}_{ip}(t,r+1) \geq \beta^\nu /n$ where $p$ belongs to $P_r^i.$

\end{proof}	
\subsection{Analysis of Algorithm~\ref{alg:fault}}
\label{sec:f_alg_analysis}

In this section we present the analysis of Algorithm~\ref{alg:fault} under Assumption~\ref{as:f_multiple_source} which does not require the network topology to achieve distributed consensus. We make the following assumption on agents' capacity to identify the true state of the world based on the Kullback-Leiber divergence between the true state's marginal $l_j(.|\theta^*)$ and marginal of any other state $l_j(.|\theta).$ The Kullback-Leiber divergence is defined as:
\begin{equation*}
   D(l_j(.|\theta^*)||l_j(.|\theta) = \sum_{\omega_i \in \calS_j} l_j(\omega_i|\theta^*) \log \frac{l_j(\omega_i|\theta^*)}{l_j(\omega_i|\theta)}.
\end{equation*}
 
\begin{assumption}
	Let $\P_\calH$ be the set of all source components in any reduced graph $\calH$ of $G(\calV,\calE).$ Then, for any $\theta \neq \theta^*,$
	for every source component $P \in \P_\calH$ for every reduced graph $\P_\calH$ the following holds:
	\[  \sum_{j \in P} D(l_j(.|\theta^*)|| l_j(.|\theta) ) \neq 0.  \]
	\label{as:fault_source}
\end{assumption}  

Intuitively, Assumption~\ref{as:fault_source} states that in any reduced graphs all agents in any source component can collaboratively detect the true state.
Before presenting our main result we define few notations from \cite{Su16a} which will be used to prove our main result. 
For each $\theta \in \Theta$ and $i \in \calV$ define $H_i(\theta,\theta^*)$ as:
\begin{align}
H_i(\theta,\theta^*) &\triangleq \sum_{\omega_i \in \calS_i} \ell_i(\omega_i|\theta^*) \log \frac{\ell_i(\omega_i|\theta)}{\ell_i(\omega_i|\theta^*)} \nonumber \\
&= -D(\ell_i(.|\theta^*) || \ell_i(.|\theta)) 
\leq 0. \label{eq:fault_h_negative}
\end{align}
Let $\calH$ be any arbitrary reduced graph with a set of source components $\P_\calH$ and $\P = \cup_{\calH \in \calR_\calF} \P_\calH$ be the set of all possible source components for all the reduced graph. Then we define $C_0$ and $C_1$ as:
\begin{align}
-C_0 &\triangleq \min_{i \in \calV} \min_{\theta_1,\theta_2 \in \Theta; \theta_1 \neq \theta_2} \min_{\omega_i \in \calS_i} \left( \log \frac{\ell_i(\omega_i|\theta_1)}{\ell_i(\omega_i|\theta_2)}  \right), \\
C_1 &\triangleq \min_{P \in \P} \min_{\theta,\theta^*\in \Theta; \theta \neq \theta^*} \sum_{i \in P} D(\ell_i(.|\theta^*)|| \ell_i(.|\theta)). \label{eq:fault_c0_c1}
\end{align}
Due to finiteness of $\Theta$ and $\calS_i$ for each agent $i,$ we know that $C_0 < \infty$ and $C_0 \geq 0.$ Also under Assumption~\ref{as:fault_source} we get $C_1 > 0.$ Since the support of $\ell_j(.|\theta)$ is the whole signal space $\calS_j$ for each $j \in \calV,$ it is easy to observe that
\begin{align}
0\ge H_j(\theta, \theta^*) &\geq \min_{w_j\in \calS_j} \pth{\log \frac{\ell_j(w_j|\theta)}{\ell_j(w_j|\theta^*)}} \nonumber \\ 
&\geq -C_0 >-\infty.
\label{eq:f_h_j_finite}
\end{align}

The following lemma is used to prove our main result.
\begin{lemma} %\cite{Su16a}
	Under Assumption~\ref{as:fault_source}, for Algorithm~\ref{alg:fault} the following statement is true for any $\theta \neq \theta^*:$
	\begin{multline*}
	\frac{1}{t^2} \sum_{r=1}^{t} \left( \sum_{j=1}^{n-\phi} {\bf \Phi}_{ij}(t,r+1) \sum_{k=1}^{r} \calL_k^j(\theta,\theta^*) \right.
	\\ \left. -r \sum_{j=1}^{n-\phi} {\bf \Phi}_{ij}(t,r+1) H_j(\theta,\theta^*) \right) \toas 0.
%	\label{eq:fault_L_H_convergence}
	\end{multline*}
	\label{lm:f_L_H_convergence}
\end{lemma}
\begin{proof}
	The lemma statement is similar (but not identical) to Lemma~3 of \cite{Su16a}. The proof of Lemma~3 of \cite{Su16a} requires each row of ${\bf \Phi}$ to converge to an identical stochastic vector. We do not have this requirement; moreover under Assumption~\ref{as:f_multiple_source} a row of ${\bf \Phi}(t,r+1)$ may not converge as $t$ goes to infinity. The proof is presented in Appendix~\ref{app:f_L_H_convergence}.
\end{proof}

Now we present our main result for non-Bayesian learning when some agents in the network are faulty.

\begin{theorem}
	Under Assumption~\ref{as:fault_source}, for Algorithm~\ref{alg:fault} every agent $i$ will concentrate its vector on the true state $\theta^*$ almost surely, i.e., $\mu_t^{i}(\theta) \toas 0~ \forall \theta \neq \theta^*.$
	\label{thm:f_hypo_convergence}
\end{theorem} 
\begin{proof}
	For any $i \in \calN$ to show $\lim_{t \rightarrow \infty} \mu_t^{i} \toas 0$ for $\theta \neq \theta^*,$ it is enough to show that $\psi_t^{i}(\theta,\theta^*) \toas -\infty.$ 
	By \eqref{eq:f_h_j_finite} we know that $|\sum_{j=1}^{n-\phi}{\bf \Phi}_{ij}(t,r+1)  H_j(\theta, \theta^*)|\le C_0<\infty$ for each agent $i \in \calN.$ Note that ${\bf \Phi}(t,r+1)$ is a row stochastic matrix. Due to finiteness of $\sum_{j=1}^{n-\phi}{\bf \Phi}_{ij}(t,r+1)  H_j(\theta, \theta^*)$ by adding and subtracting $r \sum_{j=1}^{n-\phi}{\bf \Phi}_{ij}(t,r+1)  H_j(\theta, \theta^*)$ from \eqref{eq:f_psi_i}, we get,
	\begin{align}
	\psi_t^{i}(\theta,\theta^*) = \sum_{r=1}^{t} \left( \sum_{j=1}^{n-\phi} {\bf \Phi}_{ij}(t,r+1) \sum_{k=1}^{r} \calL_k^j(\theta,\theta^*) \right. \nonumber\\
	 \left. - r \sum_{j=1}^{n-\phi} {\bf \Phi}_{ij}(t,r+1) H_j(\theta,\theta^*) \right) \nonumber\\
	  + \sum_{r=1}^{t} r \sum_{j=1}^{n-\phi} {\bf \Phi}_{ij}(t,r+1) H_j(\theta,\theta^*).
	\label{eq:f_psi_h_add}
	\end{align}
    We first derive bound for the second term.
	\begin{align}
	&\sum_{r=1}^{t} r \sum_{j=1}^{n-\phi} {\bf \Phi}_{ij}(t,r+1) H_j(\theta,\theta^*) \nonumber\\ 
	&\leq \sum_{r:t-r\geq \nu} r  \sum_{j \in P_r^i} {\bf \Phi}_{ij}(t,r+1) H_j(\theta,\theta^*), \label{eq:f_sum_bd} 
	\end{align}
	where for agent $i,$ $P_r^i$ is a source component of $\calH_r$ for which the lower bound of Lemma~\ref{lm:f_matrix} holds. The above inequality holds because by \eqref{eq:f_h_j_finite}, $H_j(\theta,\theta^*) \leq 0.$ %Note that all source components of a reduced graph $\calH_r$ are disjoint by definition.
	\begin{align}
	&\sum_{r=1}^{t} r \sum_{j=1}^{n-\phi} {\bf \Phi}_{ij}(t,r+1) H_j(\theta,\theta^*) \nonumber\\
	&\leq \sum_{r:t-r\geq \nu} r  \left(\sum_{j \in P_r^i} \frac{\beta^\nu}{n} H_j(\theta,\theta^*) \right)\nonumber\\
	&\qquad\text{By Lemma~\ref{lm:f_matrix}}\nonumber\\
	&\leq  - \sum_{r:t-r\geq \nu} r  \left(\frac{\beta^\nu}{n} C_1\right)\qquad\text{by \eqref{eq:fault_c0_c1} and \eqref{eq:fault_h_negative}} \nonumber\\
	&\leq - \frac{(t-\nu)^2}{2} \frac{\beta^{\nu}}{n} C_1. \label{eq:f_h_phi_bd}
	\end{align}
	Therefore by \eqref{eq:f_psi_h_add}, \eqref{eq:f_h_phi_bd} and Lemma~\ref{lm:f_L_H_convergence}, we get
	\begin{equation*}
	  \lim_{t \rightarrow \infty} \frac{1}{t^2}\psi_t^i(\theta,\theta^*) \leq - \frac{1}{2n} \beta^{\nu} C_1.
	\end{equation*}
	Thus, $\psi_t^{i}(\theta,\theta^*) \toas -\infty$ and $\mu_t^i(\theta) \toas 0$ for all non-faulty agents and $\theta \neq \theta^*.$

\end{proof}

\section{Conclusion}
\label{sec:conclusion}

In this work, we presented an analysis of a distributed algorithm for non-Bayesian learning over multi-agent network with adversaries which is based on a weaker assumption on the underlying network than the one present in literature \cite{Su16c,Su18a}. Our analysis does not need the network to achieve consensus among all the fault-free agents. It shows that if all the agents, whose information can reach an agent, can collaboratively correctly estimate the true state of the world then the agent itself can estimate the true state. The analysis presented here proves a sufficient network topological condition and global identifiability of the network to correctly estimate the true state by all fault-free agents. It will be interesting to prove this condition also being the necessary to estimate true state in a network with adversarial agents. 

The analysis also extends to a network with no adversaries, i.e., $f=0,$ and leads to much weaker assumption on the network as compared to the one present in literature. Previous analysis in \cite{Su16a,Nedic15} for fault-free network assume that the network is strongly connected thus capable of achieving distributed consensus. The analysis of Section~\ref{sec:faulty} can be extended to fault-free network that can have more than one connected components and each connected component may not be strongly connected. 

In this work we assume a synchronous system, i.e., in each round of the algorithm every agent sends its information at the same time to all its neighbors. In future, we would like to extend this work in case of asynchronous setting. In addition to that we assume the network to be static, i.e., neighborhood of any agent is not changing over the course of execution of the algorithm. We believe that our results can be easily generalized to the dynamic networks where the network topology is changing with time.

\bibliographystyle{abbrv}
\bibliography{consensus_bib}

\appendices

\section{Proof of Lemma~\ref{lm:f_L_H_convergence}}
\label{app:f_L_H_convergence}
%\begin{proof}
	To prove Lemma~\ref{lm:f_L_H_convergence}, we will show that almost surely for any $\epsilon>0$ there exists sufficiently large $t_\epsilon$ such that for all $t \geq t_\epsilon,$
	\begin{multline}
	\frac{1}{t^2} \left| \sum_{r=1}^{t}\sum_{j=1}^{n-\phi} {\bf \Phi}_{ij}(t,r+1) \left( \sum_{k=1}^r \calL_k^j(\theta, \theta^*) \right. \right. \\ 
	\left. \left. -rH_j(\theta, \theta^*)\right) \right| \leq \epsilon.
	\label{eq:f_ell_h_epsilon}
	\end{multline}
	For ease of notations, we will represent the left hand side of \eqref{eq:f_ell_h_epsilon} by $\frac{1}{t^2} Q(1,t).$ We prove this by dividing $r$ into two ranges $r\in \{1, \cdots, \sqrt{t}\}$ and $r\in \{\sqrt{t}+1, \cdots, t\}.$
	For $r\in \{1, \cdots, \sqrt{t}\},$ we have,
	\begin{align*}
	\frac{1}{t^2}Q(1,\sqrt{t}) &\le \frac{1}{t^2} \sum_{r=1}^{\sqrt{t}} \sum_{j=1}^{n-\phi}{\bf \Phi}_{ij}(t,r+1)\pth{2rC_0}\\
	&=\frac{1}{t^2} \pth{2C_0}  \sum_{r=1}^{\sqrt{t}}r \le C_0\pth{\frac{1}{t}+\frac{1}{t^{\frac{3}{2}}}}.
	\end{align*}
	Here first inequality is due to \eqref{eq:f_h_j_finite} and finiteness of $|\ell_k^j(\theta,\theta^*)|.$
	Thus, there exists $t_\epsilon^1$ such that for all $t \geq t_\epsilon^1$, $\frac{1}{t^2}Q(1,\sqrt{t})	\le \frac{\epsilon}{2}.$
	
	As $\calL_k^j(\theta, \theta^*)$'s are i.i.d., due to Strong Law of Large Numbers, we get $\frac{1}{r}\sum_{k=1}^r \calL_k^j(\theta, \theta^*)-H_j(\theta, \theta^*)\toas 0.$ Thus for each convergent sample path, there exists $r_\epsilon$ such that for any $r \geq r_\epsilon$,
	$\left| \frac{1}{r}\sum_{k=1}^r \calL_k^j(\theta, \theta^*)-H_j(\theta, \theta^*)\right | \le \frac{\epsilon}{2}.$
	Thus for $r \geq \sqrt{t}$ there exists sufficiently large $t_\epsilon^2$ such that for all $t\geq t_\epsilon^2,$ $r \geq \sqrt{t}$ is large enough and 
	\[ \left |\frac{1}{r}\sum_{k=1}^r \calL_k^j(\theta, \theta^*)-H_j(\theta, \theta^*)\right |\le \frac{\epsilon}{2}.\]
	For all $ t \geq t_\epsilon^2$,
	\begin{align*}
	\frac{1}{t^2}Q(\sqrt{t},t) &\le \frac{1}{t}\sum_{r=\sqrt{t}+1}^{t} \sum_{j=1}^{n-\phi}{\bf \Phi}_{ij}(t,r+1)\frac{r}{t} \frac{\epsilon}{2}\\
	&=\frac{1}{t}\sum_{r=\sqrt{t}+1}^{t} \frac{r}{t} \frac{\epsilon}{2}=\frac{\epsilon}{2}\frac{1}{t^2} \sum_{r=\sqrt{t}+1}^{t}r \\
	&=\frac{\epsilon}{4}\frac{1}{t^2} \pth{t^2-\sqrt{t}}\le ~\frac{\epsilon}{2}.
	\end{align*}

	Therefore, for every convergent path for any $\epsilon>0$, there exists $t_\epsilon =\max \{t_\epsilon^1, t_\epsilon^2\}$, such that for any $t \geq t_\epsilon$, $\frac{1}{t^2} Q(1,t) \leq \epsilon.$
	Thus \eqref{eq:f_ell_h_epsilon} holds almost surely and Lemma~\ref{lm:f_L_H_convergence} is proved.
%\end{proof}

\end{document}